\documentclass[12pt]{article}
\usepackage{amsmath, amssymb, amsthm}
\usepackage{geometry}
\usepackage{bm}
\usepackage{color}
\usepackage[normalem]{ulem}
\usepackage{graphicx}
 \usepackage{booktabs} 
 \usepackage{natbib}
\newtheorem{proposition}{Proposition}
\newtheorem{theorem}{Theorem}

\newtheorem{remark}{Remark}

\newcommand{\argmin}{\operatornamewithlimits{argmin}}

\newcommand{\bA}{\boldsymbol{A}}
\newcommand{\bc}{\boldsymbol{c}}
\newcommand{\bD}{\boldsymbol{D}}
\newcommand{\bK}{\boldsymbol{K}}
\newcommand{\bL}{\boldsymbol{L}}
\newcommand{\bP}{\boldsymbol{P}}
\newcommand{\bQ}{\boldsymbol{Q}}
\newcommand{\LOCR}{\bL_{\text{\tiny{OCR}}}}

\newcommand{\m}{\bold{m}}
\newcommand{\mOCR}{\m_{\text{\tiny{OCR}}}}

\newcommand{\Cov}{\mathrm{Cov}}
\newcommand{\Var}{\mathrm{Var}}

\newcommand{\E}{\mathbb{E}}

\newcommand{\thetaols}{\theta^{\text{\tiny{OLS}}}}
\newcommand{\thetaocr}{\theta^{\text{\tiny{OCR}}}}

\newcommand{\thetahatols}{\widehat{\theta}^{\text{\tiny{OLS}}}}
\newcommand{\thetahatocr}{\widehat{\theta}^{\text{\tiny{OCR}}}}

\newcommand{\bnu}{\boldsymbol{\nu}}

\newcommand{\hocr}{\boldsymbol{h}_0^{\text{\tiny{OCR}}}}
\newcommand{\hols}{\boldsymbol{h}_0^{\text{\tiny{OLS}}}}
\newcommand{\X}{\boldsymbol{X}}
\newcommand{\y}{\boldsymbol{y}}
\newcommand{\w}{\boldsymbol{w}}
\newcommand{\I}{\boldsymbol{I}}

\newcommand{\e}{\boldsymbol{e}}
\newcommand{\bfbeta}{\boldsymbol{\beta}}

\newcommand{\betapop}{\bfbeta^\ast_{\text{\tiny{OLS}}}}

\newcommand{\sigsqOCR}{\widehat{\sigma}^2_{\text{\tiny{OCR}}}}
\newcommand{\sigOCR}{\widehat{\sigma}_{\text{\tiny{OCR}}}}

\newcommand{\bOLS}{\widehat{\bfbeta}^{\text{\tiny{OLS}}}}
\newcommand{\bOCR}{\widehat{\bfbeta}^{\text{\tiny{OCR}}}}

\newcommand{\etaols}{\eta_{\text{\tiny{OLS}}}}

\newcommand{\etahatols}{\widehat{\eta}_{\text{\tiny{OLS}}}}
\newcommand{\etahatocr}{\widehat{\eta}_{\text{\tiny{OCR}}}}

\newcommand{\One}{\bold{1}}

\newcommand{\bmu}{\boldsymbol{\mu}}
\newcommand{\muX}{\bmu_{\scriptscriptstyle X}}
\newcommand{\muY}{\mu_{\scriptscriptstyle Y}}

\newcommand{\SigXX}{\Sigma_{\scriptscriptstyle XX}}
\newcommand{\SigXY}{\Sigma_{\scriptscriptstyle XY}}

\newcommand{\SigY}{\sigma_{\scriptscriptstyle Y}}
\newcommand{\SigW}{\sigma_{\scriptscriptstyle W}}

\newcommand{\Yhatols}{\widehat{Y}^{\text{\tiny{OLS}}}}
\newcommand{\Yhatocr}{\widehat{Y}^{\text{\tiny{OCR}}}}

\newcommand{\yhatols}{\widehat{\y}^{\text{\tiny{OLS}}}}
\newcommand{\yhatocr}{\widehat{\y}^{\text{\tiny{OCR}}}}

\title{Outcome-Calibrated Regression and Predicted Outcome-Based Inference}
\author{Hwiyoung Lee$^{1,2}$\thanks{email: hwiyoung.lee@som.umaryland.edu}
\quad
Shuo Chen$^{1,2}$\thanks{email: shuochen@som.umaryland.edu}\\[0.5em]
$^{1}$ Department of Epidemiology and Public Health, \\University of Maryland School of Medicine\\
$^{2}$ The University of Maryland Institute for Health Computing (UM-IHC)
}
\date{}

\begin{document}
\maketitle

\begin{abstract}
   Regression is a fundamental tool in scientific research. Ordinary least squares (OLS), one of the most widely used regression methods, enjoys several desirable properties, including the best linear unbiased estimator (BLUE) property. It is well known that, under the assumptions of the standard model, the OLS is conditionally unbiased given the covariates, i.e., \(\mathbb{E}(\widehat Y-Y\mid X=x)
   =0\). However, an often-overlooked property of OLS is that the prediction error is generally not unbiased conditional on the outcome, i.e., \(\mathbb{E}(\widehat Y-Y\mid Y=y)
   \neq 0\). As a consequence of minimizing mean squared error, OLS predictions are systematically shrunk toward the outcome mean,
   which explains the classical phenomenon of regression to the mean (RTM): large outcome values tend to be underpredicted, whereas small outcome values tend to be overpredicted. This conditional prediction bias creates a nonignorable problem for predicted outcome-based inference, where scientific inference is performed using the predicted outcome \(\widehat Y\) and another variable \(W\). In applications such as brain-age analysis and causal inference, we show that inference based on regression-predicted outcomes can be systematically biased. To address this issue, we propose outcome-calibrated regression (OCR), a new regression framework with a closed-form solution that directly enforces outcome calibration. The proposed OCR estimator eliminates conditional prediction bias with respect to the outcome and enables valid inference using regression-predicted outcomes.
\end{abstract}
\section{Introduction}

Regression analysis is the cornerstone of statistical methodology, permeating virtually every domain of empirical science \citep{hastie2009elements,  montgomery2021introduction}. From biomedical research and epidemiology to economics, social science, and engineering, investigators routinely fit regression models to quantify associations, generate predictions, and draw causal conclusions from data. Among regression methods, ordinary least squares (OLS) holds a uniquely prominent position: it is computationally simple, broadly interpretable, and enjoys strong theoretical guarantees. Under the classical Gauss--Markov conditions, OLS produces the best linear unbiased estimator (BLUE) of the regression coefficients, achieving minimum variance among all linear unbiased estimators, a result that has made OLS the default choice across generations of statistical practice. 

In this paper, we identify and study a fundamental yet largely overlooked property of OLS predictions. Under standard linear model assumptions \citep{seber2003linear}, OLS predictions are well known to be conditionally unbiased given the covariates, in the sense that $\E(\widehat{Y} - Y \mid X = x) = 0$. In contrast, much less attention has been paid to prediction bias conditional on the outcome itself. We show that OLS predictions can be systematically biased among observations with the same true outcome value; that is, $\E(\widehat{Y} - Y \mid Y = y) \neq 0$. This conditional prediction bias arises naturally from mean squared error minimization and causes OLS predictions to shrink toward the outcome mean. It manifests as the classical regression-to-the-mean phenomenon, in which large outcomes are systematically underpredicted and small outcomes are systematically overpredicted. While this shrinkage is well recognized as a descriptive property of OLS, its consequences for downstream statistical inference based on regression-predicted outcomes have not been fully characterized. We show that when predicted outcomes $\widehat{Y}$ are used as inputs to further inference, a practice that arises naturally in brain-age gap analysis, polygenic risk score association studies, and causal effect estimation with predicted exposures or outcomes, conditional prediction bias propagates into the inferential procedure, producing systematically attenuated associations and distorted effect estimates.

\begin{figure}[htbp]
    \centering
    \includegraphics[width=0.7\textwidth]{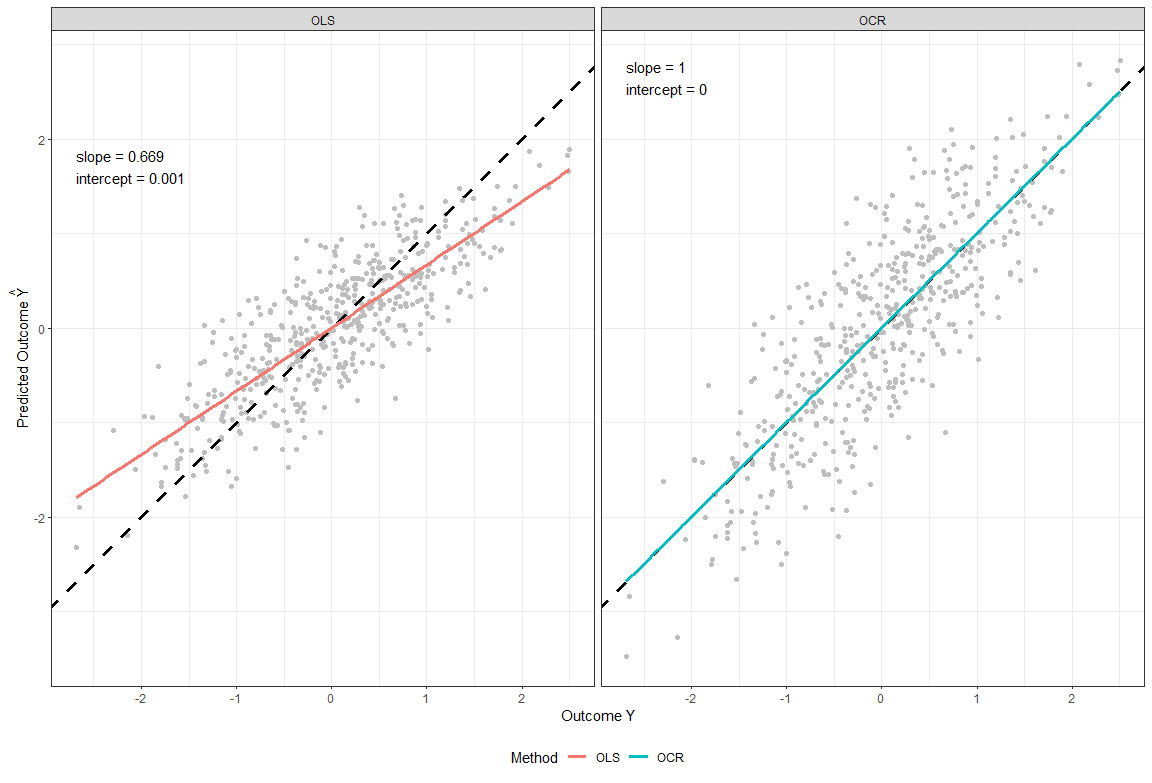}
    \caption{Overview:  scatter plots showing the relationship between  $\widehat{Y}$ vs.$Y$. The dashed line indicates unbiased outcome prediction $\widehat{Y}=Y$. The OLS predicted outcomes (left) are systematically biased. In contrast, the predicted outcomes (right) are unbiased outcome-calibrated regression. }
    \label{fig:ols_ocr}
\end{figure}

\subsection{Outcome Prediction Bias in Regression}

We begin by characterizing the conditional prediction bias of the population OLS predictor. Let $(X,Y)$ denote a pair of random variables, where $X\in\mathbb{R}^p$ is a $p$ dimensional predictor vector with $\E(X)=\muX$ and $\Var(X)=\SigXX$, and $Y\in\mathbb{R}$ is a scalar outcome with $\E(Y)=\muY$ and $\Var(Y)=\SigY^2$. Let $\Cov(X,Y)=\SigXY$, then the population OLS coefficient is 
\begin{align*}
    \betapop = \SigXX^{-1}\SigXY,
\end{align*}
and the corresponding OLS prediction can be written as
\begin{align*}
    \Yhatols = \muY + (X-\muX)^\top\betapop.
\end{align*}
We are interested in the conditional prediction bias given the true outcome value,
$\E(\widehat{Y}-Y \mid Y=y)$, which measures the average prediction error among observations with true outcome value $y$. 
\begin{equation}
\begin{aligned}
    \E(\Yhatols \mid Y=y) &= \E(\Yhatols \mid Y=y) \\
    &= \muY + \E(X-\muX\mid Y=y)^\top \betapop \\
    &= \muY + \underbrace{\frac{\SigXY^\top\SigXX^{-1}\SigXY}{\SigY^2}}_{\text{Define}\  \etaols} (y-\muY) =  \muY + \etaols(y-\muY)
\end{aligned}        
\label{eq:ols_shrinkage}
\end{equation}
Equivalently, since $\etaols$ is the population slope from regressing the OLS prediction $\Yhatols$ on the true outcome $Y$, which equals the population $R^2$. This shows that whenever $\etaols<1$, the OLS prediction is systematically shrunk toward $\muY$. The conditional prediction bias is
\begin{align*}
    \E(\Yhatols-Y \mid Y=y)  = (\etaols-1)(y-\muY).
\end{align*}    
Specifically, when $\etaols<1$, the bias is negative for outcomes above the mean and positive for outcomes below the mean:
\begin{align*}
    \begin{cases}
        \E(\Yhatols-Y\mid Y=y)<0, \quad \text{for} \ y>\muY\\
        \E(\Yhatols-Y\mid Y=y)>0, \quad \text{for} \ y<\muY.
    \end{cases}
\end{align*}
Therefore, unless $\etaols=1$, OLS systematically underpredicts the high outcomes, and overpredicts the low outcomes, exhibiting regression-to-the-mean bias.

\subsection{Biased Inference Based on Regression Predicted Outcomes }
Let   $W \in \mathbb{R}$ be a
variable associated with outcome $Y$. The target quantity is the true regression coefficient of $Y$ on $W$:
\begin{equation}
   \theta
    = \frac{\Cov(Y,\, W)}{\Var(W)}.
    \label{eq:gamma_true}
\end{equation}

However, in many applications, such as biological age estimation in aging clocks and potential outcome estimation in causal inference, $Y$ is not directly observed. Instead, inference is based on a regression-predicted outcome. The following proposition shows how OLS shrinkage propagates to downstream regression inference based on $\Yhatols$.
\begin{proposition}[Downstream bias in regressing $\Yhatols$ on $W$]
Under the OLS prediction shrinkage \eqref{eq:ols_shrinkage}, we assume that the $W$ is associated with the component of $\Yhatols$ only through its association with $Y$, that is $W$ is independent of the component of $\Yhatols$ that is not explanined by $Y$, $\Cov(\Yhatols - \E(\Yhatols\mid Y), W) = 0$. Then the regression coefficient of $\Yhatols$ on $W$ is
\begin{equation}
    \thetaols
    = \etaols\,\theta.
    \label{eq:gamma_ols_general}
\end{equation}
\end{proposition}

\begin{proof}
Under the OLS prediction setting, the population regression coefficient of
$\Yhatols$ on $W$ is:
\begin{equation}
    \thetaols
    = \frac{\Cov(\Yhatols,\, W)}{\Var(W)} = \frac{\Cov(Y,\, W)}{\Var(W)} + \frac{\Cov(\Yhatols - Y,\, W)}{\Var(W)}
    \label{eq:gamma_ols_def}
\end{equation}
Applying the law of total covariance to the second term:
\begin{align}
    \Cov(\Yhatols - Y,\, W)
    &= \E\bigl[\Cov(\Yhatols - Y,\, W \mid Y)\bigr]
     + \Cov\bigl(\E[\Yhatols - Y \mid Y],\, \E[W \mid Y]\bigr) \notag \\
    &=  \Cov\bigl((\etaols - 1)(Y - \mu_Y),\, \E[W \mid Y]\bigr) \notag \\
    &= (\etaols - 1)\,\Cov(Y,\, W),
    \label{eq:cov_resid}
\end{align}
Substituting \eqref{eq:cov_resid} into \eqref{eq:gamma_ols_def} yields \eqref{eq:gamma_ols_general}.
\end{proof}

Thus, the bias $\theta-\thetaols=(1-\etaols)\theta$ is induced by the shrinkage of $\Yhatols$ towards $\mu_Y$. In practice, this biased association estimate can lead to misguided conclusions about treatment effectiveness and risk factors associated with accelerated molecular aging.

\section{Methods: Outcome-calibrated regression (OCR)}\label{sec:sc_lr}

We consider a regression with training data sample $\mathcal{D}=(\X,\y)$, where $\X= (\X_1,\cdots,\X_n)^\top \in\mathbb{R}^{n\times p}$ denotes the matrix of $p$ predictors and $\y = (y_1,\cdots,y_n)^\top \in \mathbb{R}^n$ denotes the outcome of interest. For a linear model, let $\widehat{\y} = \X\widehat{\bfbeta} =  (\widehat{y}_1,\cdots,\widehat{y}_n)^\top \in \mathbb{R}^n$ denote the vector of predicted outcomes.
We define the empirical calibration regression of fitted values on observed outcomes as
\begin{align*}
    \widehat{y}_i=\widehat{\alpha} + \widehat{\eta}y_i + \varepsilon_i,
\end{align*}
where calibration slope and intercept are
\begin{align*}
    \widehat{\eta} &= \frac{\Cov(\y,\widehat{\y})}{\Var(\y)}=\frac{\sum_{i=1}^n (y_i-\overline{y})(\widehat{y}_i-\overline{\widehat{y}})}{\sum_{i=1}^n (y_i-\overline{y})^2},\\
    \widehat{\alpha} &= \overline{\widehat{y}}-  \widehat{\eta} \overline{y}.
\end{align*}
Ideally, a well calibrated prediction should satisfy $\widehat{\eta}=1$ and $\widehat{\alpha}=0$.
Specifically, $\widehat{\eta}=1$ means that the covariance between the observed outcomes and predicted values equals the variance of the observed outcomes, i.e., $\Cov(\y,\widehat{\y}) = \Var(\y)$, equivalently $\Cov(\y,\widehat{\y}-\y) = 0$. Thus the prediction error is uncorrelated with the observed outcome, which implies the error does not systematically increase or decrease linearly with $\y$. Under a linear model, this can be written explicitly as
\begin{align*}
    \widehat{\eta}=1 \Rightarrow \quad 
    &(\y-\overline{y}\One)^\top (\widehat{\y}-\overline{\widehat{y}}\One) = (\y-\overline{y}\One)^\top  (\y-\overline{y}\One)\\
    &(\y-\overline{y}\One)^\top \X\widehat{\bfbeta} = 
    (\y-\overline{y}\One)^\top\y \quad\quad \because (\y-\overline{y}\One)^\top \One = 0
\end{align*}
After imposing the unit slope condition, the zero intercept condition $\widehat{\alpha}=0$ is equivalent to matching the mean of the fitted values to the mean of the observed outcomes:
\begin{align*}
    \frac{1}{n}\One^\top\X\widehat{\bfbeta} = \overline{y}.
\end{align*}
Thus, the two calibration conditions can be written as linear equality
constraints in $\bfbeta$:
\begin{align*}
\begin{cases}
\widehat{\eta}(\bfbeta)=1\\
\widehat{\alpha}(\bfbeta)=0    
\end{cases}
\Longrightarrow
    \underbrace{\begin{pmatrix}
        (\y-\overline{y}\One)^\top \X\\
         \frac{1}{n}\One^\top\X
    \end{pmatrix}}_{\bA}
    \bfbeta=
      \underbrace{\begin{pmatrix}
        (\y-\overline{y}\One)^\top\y \\
        \overline{y}
    \end{pmatrix}
    }_{\bc}
\end{align*}
Therefore, we impose these calibration constraints directly  during model fitting. The
outcome calibrated linear regression (OCR) estimator is defined as
\begin{align}
\bOCR = \argmin_{\bfbeta} \frac{1}{2n}\|\y-\X\bfbeta\|_2^2 \quad \mathrm{s.t.} \ \bA\bfbeta = \bc.
\label{eq:SC_LM}
\end{align}
The Lagrangian is
\begin{align*}
    \mathcal{L}(\bfbeta,\bnu) = \frac{1}{2n}\|\y-\X\bfbeta\|_2^2  + \bnu^\top (\bA\bfbeta -\bc).
\end{align*}
Take derivative with respect to $\bfbeta$, and set it to $0$,
we obtain
\begin{align*}
    \frac{\partial\mathcal{L}(\bfbeta,\bnu)}{\partial\bfbeta} = \frac{1}{n}\X^\top\X\bfbeta-\frac{1}{n}\X^\top\y + \bA^\top\bnu = 0
\end{align*}
Therefore, the KKT system is
\begin{align*}
    \begin{bmatrix}
        \frac{1}{n}\X^\top\X & \bA^\top\\
        \bA & 0
    \end{bmatrix}
        \begin{bmatrix}
           \bfbeta\\
            \bnu
    \end{bmatrix}=
    \begin{bmatrix}
        \frac{1}{n}\X^\top\y\\
        \bc
    \end{bmatrix}
\end{align*}
From the first row 
\begin{align*}
    \bfbeta = (\X^\top\X)^{-1} (\X^\top\y-\bA^\top\bnu)
\end{align*}
Plugging this into the second row, i.e., $\bA\bfbeta=\bc$, and solving for $\bnu$ gives
\begin{align*}
    \bnu = (\bA  (\X^\top\X)^{-1} \bA^\top)^{-1} (\bA (\X^\top\X)^{-1}\X^\top\y -\bc).
\end{align*}
Substituting this expression back into $\bfbeta = (\X^\top\X)^{-1} (\X^\top\y-\bA^\top\bnu)$, then we obtain the closed form estimator:
\begin{align*}
\bOCR &= (\X^\top\X)^{-1} (\X^\top\y-\bA^\top (\bA  (\X^\top\X)^{-1} \bA^\top)^{-1} (\bA (\X^\top\X)^{-1}\X^\top\y -\bc))\\
&= \bOLS- \underbrace{(\X^\top\X)^{-1} \bA^\top  (\bA  (\X^\top\X)^{-1} \bA^\top)^{-1}}_{\bK \in \mathbb{R}^{p\times 2}} (\bA\bOLS-\bc)\\
&= \bOLS - \bK(\bA\bOLS-\bc).
\end{align*}
Here, $\bK$ is a calibration correction matrix, that maps the constraint violation of the OLS estimator, $\bA\bOLS-\bc$, into a coefficient adjustment that enforces the calibration constraints. Specifically, $\bK$ acts as a right inverse of $\bA$, i.e., $\bA\bK = \I$. Therefore, the adjustment term exactly removes the constraint violation:
\begin{align*}
    \bA\bOCR &= \bA\bOLS - \bA\bK(\bA\bOLS-\bc)=\bc.
\end{align*}

Let $\bP = \bK\bA$ and $\bQ=\I-\bP$, then we can express OCR estimator as
\begin{align*}
\bOCR = (\I-\bP)\bOLS + \bK\bc = \bQ\bOLS + \bK\bc.
\end{align*}

It also should be noted that minimizing the sum of squared errors over $\mathcal{F}$ is equivalent to choosing the feasible coefficient vector whose fitted prediction values are closest to the OLS predicted values:
\begin{align*}
    \bOCR = \argmin_{\bfbeta\in\mathcal{F}}\| \X\bfbeta-\X\bOLS\|_2^2
\end{align*}
\begin{align*}
    \bOCR &= \argmin_{\bfbeta\in\mathcal{F}}\| \y-\X\bfbeta\|_2^2 \\
    &= \argmin_{\bfbeta\in\mathcal{F}} \| \y-\X\bOLS + \X\bOLS - \X\bfbeta\|_2^2\\
    &= \argmin_{\bfbeta\in\mathcal{F}}\| \y-\X\bOLS\|_2^2 + \|\X\bOLS - \X\bfbeta\|_2^2 + 2(\y-\X\bOLS)^\top\X(\bOLS-\bfbeta)\\
    &= \argmin_{\bfbeta\in\mathcal{F}} \| \y-\X\bOLS\|_2^2 + \|\X\bOLS - \X\bfbeta\|_2^2 \\
    &=\argmin_{\bfbeta\in\mathcal{F}}\| \X\bfbeta-\X\bOLS\|_2^2.
\end{align*}

\begin{remark}
   The outcome calibrated regression (OCR) estimator  $\bOCR$ is the orthogonal projection of $\bOLS$ onto the calibration constraints set $\mathcal{F}$ under the $\X^\top\X$-weighted inner product, i.e., $\langle u,v\rangle_{\X^\top\X} = u^\top\X^\top\X v$
\end{remark}

\subsection{Properties of OCR}
We consider a regression task with training data $\mathcal{D} = (\X, \bm{y})$, where
$\X \in \mathbb{R}^{n \times p}$ is the design matrix and $\y \in \mathbb{R}^n$ is the
outcome vector. The linear model is $\bm{y} = X\bm{\beta} + \bm{\varepsilon}$, with
$\mathbb{E}(\bm{\varepsilon}) = \bm{0}$ and $\Var(\bm{\varepsilon}) = \sigma^2 \I$.

\paragraph{Minimum Variance within $\mathcal{F}$.}

\begin{theorem}[Constrained Gauss-Markov]
    Among all linear unbiased estimators $\widetilde{\bfbeta} = \bL\y+\m$ for $\bfbeta \in \mathcal{F}$, then $\bOCR$ has minimum variance, in the sense  that
    \begin{align*}
       \Var(\widetilde{\bfbeta}) - \Var(\bOCR)
         \succeq  0.
    \end{align*}
\end{theorem}
\begin{proof}
    We express $\bOCR = (\I - \bK\bA)\bOLS + \bK\bc = \underbrace{(\I-\bK\bA)(\X^\top\X)^{-1}\X^\top}_{\LOCR}\y + \bK\bc$. So
    the difference is 
    \[\widetilde{\bfbeta}-\bOCR = (\bL-\LOCR)\y + (\m-\mOCR).\]
    We define $\bD = \bL-\LOCR$ and $\e = \m-\mOCR$
    Since both $\widetilde{\bfbeta}$ and $\bOCR$ are unbiased for $\bfbeta\in\mathcal{F}$,
    \[\E(\widetilde{\bfbeta}-\bOCR\mid \X) = \bD\E(\y\mid\X) + \e = \bD\X\bfbeta + \e=0\]
    We have $\widetilde{\bfbeta}=\bOCR+(\widetilde{\bfbeta}-\bOCR)$.   
    \begin{align*}
        \Var(\widetilde{\bfbeta}) &= \Var(\bOCR) + \Var(\widetilde{\bfbeta}-\bOCR) + 2\Cov(\bOCR,\widetilde{\bfbeta}-\bOCR)
    \end{align*}
We can express the cross term as $\Cov(\bOCR,\widetilde{\bfbeta}-\bOCR) = \sigma^2\LOCR\bD^\top$.
Since $\bA\bK=\I$, we have $(\I-\bK\bA)(\X^\top\X)^{-1} \in \text{Null}(\bA)$, which implies $\bD\X\{(\I-\bK\bA)(\X^\top\X)^{-1}\}=0$. Thus, $\LOCR\bD^\top = (\I-\bK\bA)(\X^\top\X)^{-1}\X^\top\bD^\top = (\bD\X\{(\I-\bK\bA)(\X^\top\X)^{-1}\})^\top = 0$. Then we have
    \begin{align*}
       \Var(\widetilde{\bfbeta}) - \Var(\bOCR)&= \Var(\widetilde{\bfbeta}-\bOCR) =\sigma^2\bD\bD^\top  \succeq 0\\
    \end{align*}
\end{proof}

\paragraph{Unbiased association analysis for $\Yhatocr$ and $W$. }

\begin{proposition}[Downstream unbiasedness in regressing $\Yhatocr$ on $W$]
\label{prop:sc_unbiased}
Under the calibration constraint and the conditional
independence assumption $W \perp (\Yhatocr - Y) \mid Y$, which often holds as the residual $\Yhatocr - Y$ being random noise. The population regression coefficient
of $\Yhatocr$ on $W$ satisfies $\thetaocr=\theta$.
This shows that the OCR coefficient recovers the true regression coefficient of $Y$ on $W$.
\end{proposition}

\section{Inference for Outcome Calibrated Regression}
In this section, we focus on  the statistical inference for regression parameters and prediction intervals for outcome calibrated regression 

\subsection{Inference for $\bOCR$}
\begin{proposition} [Variance--Covariance of $\bOCR$]
Under the standard linear model assumptions
$\mathbb{E}[\bm{\varepsilon}] = \bm{0}$, $\mathrm{Var}(\bm{\varepsilon}) = \sigma^2 I$,
treating $\bA$ and $\bc$ as fixed (i.e.\ conditioning on $X$ and the constraint
being determined), the variance-covariance matrix of $wide\hat{\bm{\beta}}^{SC}$ is:
\begin{equation}
    \Var(\bOCR)
    = \sigma^2 \bQ(X^\top X)^{-1}\bQ^\top.
    \label{eq:var_sc}
\end{equation}
\end{proposition}

\begin{remark}
    Note that $\bP$ satisfies $\bP(\X^\top \X)^{-1}\bP^\top = \bP(\X^\top \X)^{-1}$, so \eqref{eq:var_sc} simplifies to:
    \[
        \Var(\bOCR)
        = \sigma^2\bigl[(\X^\top \X)^{-1} - (\X^\top \X)^{-1}\bA^\top
          \bigl(\bA(\X^\top \X)^{-1}\bA^\top\bigr)^{-1}\bA(\X^\top \X)^{-1}\bigr].
    \]
\end{remark}
This shows that $\Var(\bOCR) \preceq \sigma^2(\X^\top \X)^{-1} = \Var(\bOLS)$. Thus, under these constraints, OCR has coefficient estimation variance no larger than that of OLS.

The residual sum of squares under the calibration constraint is $\mathrm{RSS}_{\text{\tiny{OCR}}} = \|\y - \X\bOCR\|^2$.
Based on the restricted least squares with equality constraints $\bA\bfbeta=\bc$, where $r = \mathrm{rank}(\bA) = 2$, the parameter space has dimension $p-r$. Thus, the residual degrees of freedom is $n-p+2$. 

\begin{equation}
    \sigsqOCR = \frac{\mathrm{RSS}_{\text{\tiny{OCR}}}}{n - p+2}.
    \label{eq:sigma_sc}
\end{equation}

Using this estimate, the standard error of the $j$-th coefficient $\bOCR_j$ is
\begin{align*}
    \mathrm{SE}(\bOCR_j) = \sigOCR \sqrt{\bigl[(\X^\top \X)^{-1} - (\X^\top \X)^{-1}\bA^\top
      \bigl(\bA(\X^\top \X)^{-1}\bA^\top\bigr)^{-1}\bA(\X^\top \X)^{-1}\bigr]_{jj}},
    \label{eq:var_sc_est}
\end{align*}
where $[\cdot]_{jj}$ denotes the $j$-th diagonal entry.
A $(1-\alpha)$ confidence interval for $\bfbeta_j$ is then given by
\begin{align*}
    \bOCR_j \pm t_{n-p+2,1-\alpha/2} \mathrm{SE}(\bOCR_j),
\end{align*}
where $t_{n-p+2,\alpha/2}$ is the $1-\alpha/2$ quantile of the $t$-distribution with $n-p$ degrees of freedom. Wald-type tests for individual coefficients can be performed using the corresponding test statistic, $\bOCR_j/\mathrm{SE}(\bOCR_j)$.

\subsection{Confidence and Prediction Intervals}
\paragraph{Confidence interval for the mean response.}

For a covariate vector $\bm{x}_0 \in \mathbb{R}^p$, the OCR predicted mean response is $\yhatocr_0 = \bm{x}_0^\top \hat{\bm{\beta}}^{SC}$.
Using the variance for $\bOCR$, the variance of the fitted response is $\Var(\yhatocr_0)=\sigma^2\hocr$, where
\begin{equation}
\hocr = \bm{x}_0^\top
      \bigl[(\X^\top \X)^{-1} - (\X^\top \X)^{-1}\bA^\top
      (\bA(\X^\top \X)^{-1}\bA^\top)^{-1}\bA(\X^\top \X)^{-1}\bigr]
      \bm{x}_0.
\end{equation}
Replacing $\sigma^2$ by $\sigsqOCR$, the $(1-\alpha)$ confidence interval for the mean response $\E(Y\mid\bm{x}_0)=\bm{x}_0^\top\bm{\beta}$ is 
\[
\yhatocr_0 \pm t_{n-p+2,1-\alpha/2} \sigOCR \sqrt{\hocr}
\] 

This interval quantifies the uncertainty in estimating the population mean response at $\bm{x}_0$. Its width is governed by the OCR leverage $\hocr$, and increases for covariate values with a higher leverage under the OCR design. In OLS, the analogous interval is narrowest at the covariate mean ($\bm{x}_0 = \bar{\bm{x}}$), but this property is not guaranteed under OCR because the calibration constraints modify the leverage structure.\\

\begin{remark}
    Compared to the OLS confidence interval, which uses the leverage
    $\hols = \bm{x}_0^\top(\X^\top \X)^{-1}\bm{x}_0$, the OCR confidence interval uses the constrained leverage $\hocr$. Treating the calibration constraints as fixed, the OCR leverage satisfies \[\hocr \leq \hols\], because the OCR covariance matrix subtracts a positive semidefinite adjustment from the OLS covariance matrix. This reflects the variance reduction induced by restricting estimation to the calibrated parameter space $\mathcal{F}$. Thus, conditional on the calibration constraints, the fitted mean response has no larger variance than its OLS counterpart. 
\end{remark}

\paragraph{Prediction interval for a new observation.}
For a new observation at covariate vector $\bm{x}_0$, assume $y_0^{\mathrm{new}}= \bm{x}_0^\top\bfbeta + \varepsilon_0$ with $\varepsilon_0 \sim N(0, \sigma^2)$ independent of the training data. The prediction error satisfies $y_0^{\mathrm{new}} - \yhatocr_0 = \varepsilon_0 + \bm{x}_0^\top(\bfbeta - \bOCR).$ Treating the OCR calibration constraints as fixed, we have $\Var(y_0^{\mathrm{new}} - \yhatocr_0) = \sigma^2\bigl(1 + \hocr\bigr)$, where the additional 1 accounts for the irreducible noise in a single new observation.
Therefore, the $(1-\alpha)$ prediction interval for a new observation at $\bm{x}_0$ is
\begin{equation}
\yhatocr_0 \pm t_{n-p+2,1-\alpha/2} \sigOCR \sqrt{1+\hocr}.
   \label{eq:pi_new}
\end{equation}

\section{Simulation Study}
We perform simulation to assess the bias of outcome prediction and estimated regression parameter of $\widehat{Y}$ on $W$.

\subsection{Unbiased Calibration}

We first conducted a simulation study to evaluate the unbiased calibration
properties of outcome calibrated regression (OCR) relattive to ordinary least squares (OLS).
For each replication, we generated independent covariates
$\X \in \mathbb{R}^{n \times p}$ with $\X_{ij} \overset{\text{iid}}{\sim}
\mathcal{N}(0,1)$, and outcomes
\begin{equation}
  \y_i \;=\; \X_i^\top \boldsymbol{\beta} + \varepsilon_i,
  \qquad \varepsilon_i \overset{\text{iid}}{\sim} \mathcal{N}(0, \sigma^2),
  \label{eq:dgp}
\end{equation}
where $\boldsymbol{\beta} = (0.5, \ldots, 0.5)^\top \in \mathbb{R}^P$.
We adopted a factorial design crossing sample size $n \in \{200, 500\}$,
number of predictors $p \in \{2, 5\}$, and noise level
$\sigma \in \{0.5, 1.0\}$, yielding eight scenarios. We computed three metrics on the training sample:
\begin{enumerate}
  \item \textbf{Mean squared error (MSE):}
    $\mathrm{MSE} = n^{-1}\sum_{i=1}^n (\y_i - \widehat{\y}_i)^2$.
  \item \textbf{Calibration slope} ($\eta$): the slope from regressing
    $\widehat{\y}$ on $\y$.
  \item \textbf{Calibration intercept} ($\alpha$): the intercept from the
    same regression. 
\end{enumerate}
The results averaged over $500$ replications are provided in Table~\ref{tab:sim_results}.

\begin{figure}[ht]
  \centering
  \includegraphics[width=\textwidth]{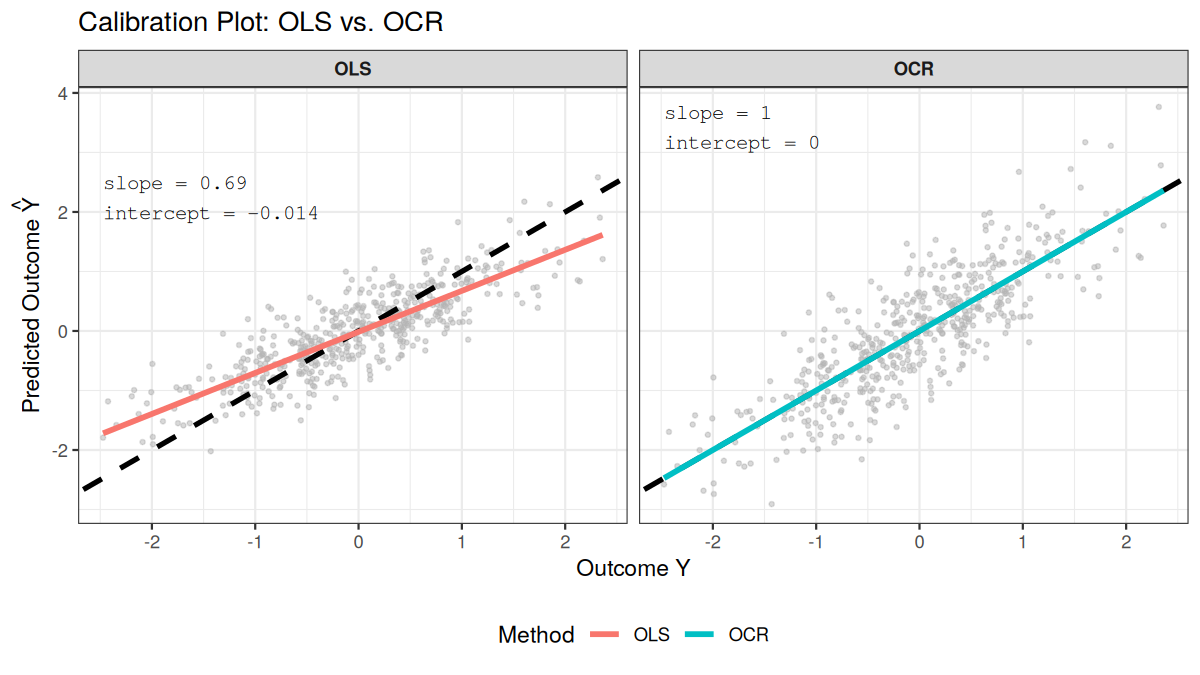}
  \caption{Representative calibration plots ($n = 500$, $p = 2$, $\sigma = 0.5$).
    Each panel plots predicted values $\widehat{\y}$ against observed outcomes $\y$.
    The dashed line is the identity ($\widehat{\y} = \y$); the solid line is the
    fitted calibration regression.
    OLS exhibits systematic slope attenuation ($\eta \approx 0.69$),
    whereas OCR achieves exact calibration ($\eta = 1$, $\alpha = 0$)
    by construction.}
  \label{fig:calib_plot}
\end{figure}

\begin{table}[ht]
\centering
\setlength{\tabcolsep}{4pt}
\caption{Simulation results ($500$ replications per scenario).
Entries are mean\,(SD) over replications.
Calibration slope~$\widehat{\eta}$ and intercept~$\widehat{\alpha}$ are
obtained by regressing $\widehat{\y}$ on $\y$;
perfect calibration requires $\widehat{\eta}=1$, $\widehat{\alpha}=0$.}
\label{tab:sim_results}
\resizebox{\textwidth}{!}{%
\begin{tabular}{ccc cc cc cc}
\toprule
& & &
\multicolumn{2}{c}{MSE} &
\multicolumn{2}{c}{Calibration Slope $\hat{\eta}$} &
\multicolumn{2}{c}{Calibration Intercept $\hat{\alpha}$} \\
\cmidrule(lr){4-5}\cmidrule(lr){6-7}\cmidrule(lr){8-9}
$n$ & $p$ & $\sigma$ & OLS & OCR & OLS & OCR & OLS & OCR \\
\midrule
200 & 2 & 0.5 & 0.247\,(0.025) & 0.374\,(0.052) & 0.666\,(0.038) & 1.000\,(0.000) & -0.001\,(0.020) & $\approx 0$\,(${<}10^{-16}$) \\
500 & 2 & 0.5 & 0.248\,(0.016) & 0.372\,(0.033) & 0.669\,(0.024) & 1.000\,(0.000) & 0.001\,(0.013) & $\approx 0$\,(${<}10^{-16}$) \\
200 & 5 & 0.5 & 0.243\,(0.024) & 0.291\,(0.035) & 0.837\,(0.021) & 1.000\,(0.000) & -0.001\,(0.014) & $\approx 0$\,(${<}10^{-16}$) \\
500 & 5 & 0.5 & 0.246\,(0.016) & 0.295\,(0.022) & 0.835\,(0.014) & 1.000\,(0.000) & 0.001\,(0.009) & $\approx 0$\,(${<}10^{-16}$) \\
\midrule
200 & 2 & 1.0 & 0.978\,(0.101) & 2.956\,(0.668) & 0.341\,(0.052) & 1.000\,(0.000) & -0.002\,(0.054) & $\approx 0$\,(${<}10^{-16}$) \\
500 & 2 & 1.0 & 0.994\,(0.064) & 3.011\,(0.432) & 0.335\,(0.034) & 1.000\,(0.000) & -0.002\,(0.034) & $\approx 0$\,(${<}10^{-16}$) \\
200 & 5 & 1.0 & 0.970\,(0.098) & 1.739\,(0.282) & 0.564\,(0.046) & 1.000\,(0.000) & 0.003\,(0.045) & $\approx 0$\,(${<}10^{-16}$) \\
500 & 5 & 1.0 & 0.992\,(0.064) & 1.787\,(0.188) & 0.558\,(0.030) & 1.000\,(0.000) & 0.001\,(0.029) & $\approx 0$\,(${<}10^{-16}$) \\
\bottomrule
\end{tabular}%
}
\end{table}


\paragraph{Results.} Across all eight scenarios, the OCR calibration slope and intercept are numerically $1.000$ and 0, respectively. This confirms that the two constraints imposed in \eqref{eq:SC_LM} are satisfied. However, the OLS calibration slope is strictly less than one in every scenario,
ranging from $0.335$ ($n=200$, $p=2$, $\sigma=1.0$) to $0.837$
($n=500$, $p=5$, $\sigma=0.5$). The noise level is the dominant driver of attenuation. Increasing $\sigma$ from $0.5$ to $1.0$ roughly halves the OLS calibration slope (e.g., $0.669 \to 0.335$ for $n=500$, $p=2$).

By enforcing the calibration constraint, OCR inflates MSE relative to OLS.
The constraint is modest under low noise ($\Delta\text{MSE} \approx 0.05$ for
$\sigma=0.5$, $p=5$) but substantial under high noise
($\Delta\text{MSE} \approx 2.07$ for $\sigma=1.0$, $p=2$).

\subsection{Estimation of $\theta$ via Regression of $\widehat{\y}$ on $\w$}

We further conducted a simulation study to examine downstream bias in the association between the predicted outcome and an external variable. To do this, we generated the pair of $(\y, \w)$ from a bivariate normal distribution. Specifically, in each replication, $(\y_i, \w_i)_{i=1}^{n}$ are drawn i.i.d.\ from
\begin{equation*}
  \begin{pmatrix} \y \\ \w \end{pmatrix}
  \sim
  \mathrm{MVN}\!\left(
    \mathbf{0},\;
    \boldsymbol{\Sigma}
  \right),
  \qquad
  \boldsymbol{\Sigma}
  =
  \begin{pmatrix}
    \SigY^2            & \theta\,\SigW^2 \\
    \theta\,\SigW^2   & \SigW^2
  \end{pmatrix},
  \label{eq:dgp_theta}
\end{equation*}
where $\SigY^2 = p\cdot(0.5)^2 + \sigma^2$ mirrors the data generating process in which $\y = \X\bfbeta+\varepsilon$ with
$\boldsymbol{\beta}=0.5\cdot\mathbf{1}_p$ and
$\varepsilon\sim\mathcal{N}(0,\sigma^2)$. Here, the estimand is $\theta$, the population regression coefficient of $\y$ on $\w$. We fix $\theta=0.5$ and $\SigW=1$ throughout. The covariate matrix $\X\in\mathbb{R}^{n\times p}$ was also generated to serve as prediction model features. For each method, $\widehat{\theta}$ is obtained by regressing $\widehat{\y}$ on $\w$.

Estimator performance is summarized across $500$ replications via four
metrics. The average estimate $\overline{\widehat{\theta}}$ and standard
deviation $\mathrm{sd}(\widehat{\theta})$ characterize the central
tendency and sampling variability of $\widehat{\theta}$, while bias is defined
as $\mathrm{Bias} = \overline{\widehat{\theta}} - \theta$. To quantify the
relative magnitude of bias with respect to sampling variability, we computeed
the bias-to-standard deviation ratio, $\mathrm{BSR} = |\mathrm{Bias}|\,/\,\mathrm{sd}(\widehat{\theta})$,
with values exceeding unity indicating that systematic bias dominates Monte
Carlo variability. Finally, empirical coverage is the proportion of
replications in which the nominal 95\% Wald interval
$\widehat{\theta}_m \pm 1.96\,\mathrm{se}(\widehat{\theta})$
contains the true $\theta$.

\paragraph{Results.} The results from Table~\ref{tab:theta_plugin}  are consistent with the theoretical results.
The OLS plug-in, $\thetahatols$, is severely biased toward zero in every scenario
($\mathrm{Bias}\approx -0.50$, $\mathrm{BCR}\in[26,180]$,
coverage $= 0\%$). This occurs because OLS predictions are slope attenuated: $\yhatols$ preserves only $\etahatols$ fraction of the association between $\y$ and $\w$. Thus $\thetahatols \approx \etahatols\theta$, which can be close to 0 when $\etahatols$ is small. However, the OCR plug-in, $\thetaocr$, is approximately unbiased across all scenarios
($|\mathrm{Bias}|<0.25$, $\mathrm{BCR}<0.08$,
coverage $\in[95.2\%,\, 98.6\%]$). This confirms that enforcing $\etahatocr=1$ yields an unbiased downstream association $\thetahatocr\approx\theta$ by eliminating attenuation induced prediction bias.

\begin{table}[ht]
\centering
\setlength{\tabcolsep}{4pt}
\caption{Simulation results for the plug in estimator $\widehat{\theta}$
from regressing $\widehat{\y}$ on $\w$ ($500$ replications; true $\theta=0.5$,
$\SigW=1$).
Entries in the Mean\,(SD) and Bias\,(SD) columns are
mean\,(standard deviation) over replications.
BSR: bias-to-standard deviation ratio $|\mathrm{Bias}|/\mathrm{SD}$;
values ${>}1$ indicate bias dominates sampling variability.
Coverage: empirical 95\% CI coverage (\%) of $\theta$.}
\label{tab:theta_plugin}
\resizebox{\textwidth}{!}{%
\begin{tabular}{ccc cc cc cc cc}
\toprule
& & &
\multicolumn{2}{c}{Mean\,(SD) of $\widehat{\theta}$} &
\multicolumn{2}{c}{Bias\,(SD)} &
\multicolumn{2}{c}{BSR} &
\multicolumn{2}{c}{Coverage (\%)} \\
\cmidrule(lr){4-5}\cmidrule(lr){6-7}\cmidrule(lr){8-9}\cmidrule(lr){10-11}
$n$ & $p$ & $\sigma$ & OLS & OCR & OLS & OCR & OLS & OCR & OLS & OCR \\
\midrule
200 & 2 & 0.5 & 0.005\,(0.007) & 0.576\,(2.202) & $-$0.495\,(0.007) & 0.076\,(2.202) & 72.588 & 0.035 & 0.0 & 98.6 \\
500 & 2 & 0.5 & 0.002\,(0.003) & 0.433\,(1.613) & $-$0.498\,(0.003) & $-$0.067\,(1.613) & 179.853 & 0.042 & 0.0 & 97.8 \\
200 & 5 & 0.5 & 0.012\,(0.014) & 0.503\,(0.620) & $-$0.488\,(0.014) & 0.003\,(0.620) & 36.111 & 0.006 & 0.0 & 97.8 \\
500 & 5 & 0.5 & 0.005\,(0.006) & 0.533\,(0.631) & $-$0.495\,(0.006) & 0.033\,(0.631) & 84.408 & 0.052 & 0.0 & 97.2 \\
\midrule
200 & 2 & 1.0 & 0.004\,(0.009) & 0.663\,(3.294) & $-$0.496\,(0.009) & 0.163\,(3.294) & 58.078 & 0.049 & 0.0 & 96.0 \\
500 & 2 & 1.0 & 0.002\,(0.004) & 0.257\,(3.122) & $-$0.498\,(0.004) & $-$0.243\,(3.122) & 135.156 & 0.078 & 0.0 & 96.4 \\
200 & 5 & 1.0 & 0.013\,(0.019) & 0.495\,(0.785) & $-$0.487\,(0.019) & $-$0.005\,(0.785) & 26.019 & 0.006 & 0.0 & 95.2 \\
500 & 5 & 1.0 & 0.005\,(0.007) & 0.495\,(0.801) & $-$0.495\,(0.007) & $-$0.005\,(0.801) & 73.918 & 0.007 & 0.0 & 97.0 \\
\bottomrule
\end{tabular}%
}
\end{table}

\section{Discussion}

In standard OLS regression, the estimator satisfies the classical unbiasedness
condition $\E(\widehat{Y} - Y \mid X = x) = 0$, which ensures a consistent estimation of the regression coefficients $\bfbeta$. However, this covariate conditional
unbiasedness does not imply outcome conditional unbiasedness. As we have shown,
the OLS predicted outcome is biased conditional on $Y$, i.e., $\E(\widehat{Y} - Y \mid Y = y) \neq 0$ in general, and the bias is systematic: predictions are shrunk towards the
population mean $\muY$ by a factor of $\etaols$, so that high outcomes are
underestimated and low outcomes are overestimated. This regression-to-the-mean
phenomenon is an inherent property of OLS, not a finite-sample artifact, and
it propagates into any downstream analysis that substitutes $\widehat{Y}$ for $Y$.
In particular, regressing a predicted outcome $\widehat{Y}$ on an external variable $W$ attenuates the
estimated association by a factor of $\etaols$ relative to the true coefficient
$\theta$, leading to systematic underestimation of the effect sizes and
potentially erroneous scientific conclusions.

To address this, we propose the Outcome Calibrated Regression (OCR) estimator,
which directly imposes the outcome-conditional unbiasedness condition
$\E(\widehat{Y} - Y \mid Y = y) = 0$. By construction, OCR eliminates the
regression-to-the-mean bias in the predicted outcome, so that inference based
on $\Yhatocr$ recovers the true regression coefficient
$\theta$ without attenuation. The implications are broad. OCR is relevant whenever predicted outcomes are used as surrogates for unobserved or expensive-to-measure ground-truth outcomes, including treatment-effect estimation in causal inference pipelines, two-sample Mendelian randomization, polygenic score association studies, aging clock and biological age analyses, and federated or transfer learning settings in which only model predictions are shared across sites. In these contexts, replacing OLS predictions with OCR-corrected predictions can eliminate systematic prediction bias and reduce bias in downstream inference attributable to RTM.
\bibliographystyle{apalike} 
\bibliography{ref}

@book{seber2003linear,
  title     = {Linear Regression Analysis},
  author    = {Seber, George A. F. and Lee, Alan J.},
  edition   = {2nd},
  year      = {2003},
  publisher = {Wiley},
  address   = {Hoboken, NJ}
}

@book{hastie2009elements,
  title     = {The Elements of Statistical Learning: Data Mining,
               Inference, and Prediction},
  author    = {Hastie, Trevor and Tibshirani, Robert and Friedman, Jerome},
  edition   = {2nd},
  year      = {2009},
  publisher = {Springer},
  address   = {New York}
}

@book{montgomery2021introduction,
  title     = {Introduction to Linear Regression Analysis},
  author    = {Montgomery, Douglas C. and Peck, Elizabeth A. and Vining, G. Geoffrey},
  edition   = {6th},
  year      = {2021},
  publisher = {Wiley},
  address   = {Hoboken, NJ}
}

\end{document}